\documentclass[3p, sort&compress]{elsarticle}
\usepackage{color}
\usepackage{amsmath}
\usepackage{amsthm}
\usepackage{amssymb}
\usepackage{titlesec}
\usepackage[thinlines,thiklines]{easybmat}

\usepackage{graphicx}
\usepackage{epsfig}
\theoremstyle{definition}\newtheorem{Df}{Definition}
\theoremstyle{plain}\newtheorem{Th}{Theorem}
\theoremstyle{definition}\newtheorem{Rm}{Remark}
\theoremstyle{definition}
\theoremstyle{plain}
\theoremstyle{plain}\newtheorem{Co}[Th]{Corollary}
\theoremstyle{plain}\newtheorem{Lm}[Th]{Lemma}

\begin{document}
\begin{frontmatter}
\title{ \Large{Power of the interactive proof systems with verifiers
                modeled by semi-quantum two-way finite automata}}

\author{Shenggen Zheng$^{1,2}$}

\author{Daowen Qiu$^{1,3,}$ \corref{one}}

\author{Jozef Gruska$^{2}$}

 \cortext[one]{Corresponding author.\\ \indent{\it E-mail address:} issqdw@mail.sysu.edu.cn (D. Qiu).}

\address{

  $^1$Department of Computer Science, Sun Yat-sen University,
Guangzhou 510006,
  China\\
   $^{2}$Faculty of Informatics, Masaryk University, Brno 60200, Czech Republic\\
$^3$SQIG--Instituto de Telecomunica\c{c}\~{o}es, Departamento de Matem\'{a}tica, Instituto Superior T\'{e}cnico,
 Universidade de Lisboa, Av. Rovisco Pais 1049-001, Lisbon, Portugal\\
}

\begin{abstract}

{\em Interactive proof systems} (IP) are very powerful~--~languages they can accept form exactly PSPACE. They represent also one of the very fundamental concepts of theoretical computing and a model of computation by interactions. One of the key players in IP are verifiers. In the original model of IP whose power is that of PSPACE, the only restriction on verifiers is that they work in randomized polynomial time. Because of such key importance of IP, it is of large interest to find out how powerful will IP be when verifiers are more restricted. So far this
was explored for the case that verifiers are {\em two-way probabilistic finite automata} (Dwork and Stockmayer, 1990) and {\em one-way quantum finite automata} as well as {\em  two-way quantum finite automata } (Nishimura and Yamakami, 2009).
IP in which verifiers uses {\em public randomization} is called {\em Arthur-Merlin proof systems} (AM).
 AM with  verifiers modeled by Turing Machines augmented with a fixed-size quantum register  (qAM) were studied  also by Yakaryilmaz (2012). He proved, for example,  that  an NP-complete language $L_{knapsack}$, representing the 0~-~1 knapsack problem,  can be recognized by a qAM whose verifier is a {\em two-way finite automaton}  working on  quantum mixed states using superoperators.

In this paper we explore the power of AM for the case that verifiers are {\em two-way finite automata with quantum and classical states} (2QCFA)~--~introduced by Ambainis and Watrous in 2002~--~and the communications are classical. It is of interest to consider AM with such ``semi-quantum" verifiers because they use only limited quantum resources. Our main result is that such Quantum Arthur-Merlin proof systems (QAM(2QCFA)) with polynomial expected running time  are  more powerful than the models in which the verifiers are two-way probabilistic finite automata (AM(2PFA)) with polynomial expected running time. Moreover, we prove that there is a language which can be recognized by an exponential expected running time QAM(2QCFA), but can not be recognized by any AM(2PFA), and that the NP-complete language $L_{knapsack}$ can also be recognized by a QAM(2QCFA) working only on quantum pure  states using unitary operators.

\end{abstract}

\end{frontmatter}

\section{Introduction}

An important way to get deeper insights into the power of various quantum resources and operations is to explore the power of various quantum variations of the basic models of classical automata. Of a special interest is to do that for various quantum variations of the classical finite automata, especially for those that use limited amounts of always expensive quantumness~--~quantum  resources: states, correlations, operations and measurements. This paper aims to contribute to such a line of research.

There are two basic approaches toward how to introduce quantum features to
classical models of finite automata. The first one is to consider quantum
variants of the classical {\em one-way (deterministic) finite automata}
(1FA or 1DFA) and the second one is to consider quantum variants of the
classical {\em two-way finite automata} (2FA or 2DFA). Already the very first
attempts to introduce such models, by Moore and Crutchfields \cite{Moo97} as well as
Kondacs and Watrous \cite{Kon97} demonstrated that in spite of the fact that in the
classical case, 1FA and 2FA have the same recognition power, this is not so for their quantum
variations (in case only unitary operations and projective measurements are considered as quantum operations).
Moreover, already the first model of {\em two-way
quantum finite automata} (2QFA), namely that introduced by Kondacs and Watrous,
demonstrated that quantum variants of 2FA are much too
powerful~--~they can recognize even some {\em non-context free languages} and
are actually not really finite  in a strong sense \cite{Kon97}. It started to be
therefore of interest to introduce and explore some ``less quantum"
variations of 2FA and their power \cite{Amb06,Amb98,Amb02,Bro99,Ber03,LiQiu09,Mat12,Mer06,Yak10,Yak11}.

A ``hybrid" quantum variation of 2FA, namely, {\em two-way
finite automata with quantum and classical states}  (2QCFA) was
introduced by Ambainis and Watrous \cite{Amb02}.  Using this model they were able
to show, in an elegant way, that already an addition of a single qubit to a
classical model can much increase its power. A 2QCFA is
essentially a classical 2FA augmented with a quantum memory of constant
size (for states of a fixed Hilbert space) that does not depend on the
size of the (classical) input. In spite of such a restriction, 2QCFA have
been shown to be more powerful than {\em two-way probabilistic finite automata}
(2PFA) \cite{Amb02,ZhgQiu11,Zhg12}.

In mid 1980s, Babai \cite{Bab85} and Goldwaser et al.  \cite
{Gol89}, independently, introduced
so-called {\em interactive proof systems} with unlimited power provers and polynomial power randomized verifiers. A famous result of \cite{Sha92}, stated as IP=PSPACE,
that languages recognized by IP are exactly those from PSPACE, demonstrated
enormous power hidden in simple interactions of IP.

It is therefore natural to explore power also of some weaker variations of IP.
Since unlimited power of provers seems to be very essential for the whole concept of IP,
the research started to focus on the cases with limited power verifiers. This has been done at first by Dwork and Stockmeyer \cite{DwS92}~--~they explored the case that verifiers are {\em two-way probabilistic finite automata} (IP(2PFA)). They showed that every language in the class EXP can be accepted by some IP(2PFA). However, the set of languages recognized by such IP in which verifiers use {\em public randomization} (also called Arthur-Merlin proof systems) is a proper subset of P.
Later, Nishimura and Yamakami \cite{Nis09} explored the case that verifiers are modeled by {\em one-way quantum finite automata} as well as {\em two-way quantum finite automata} and demonstrated strengths and weaknesses of both IP.

Of importance is also a variant of IP, called {\em Arthur-Merlin proof systems} (AM). The difference between IP and AM is that
the prover of IP has at each step only {\em partial information} of the
configuration of the verifier while the prover of AM always has {\em complete information}  of the current
configuration of the verifier.
Also for such {\em interactive proof systems} it is of importance to explore their
power for the case that verifiers have a more limited power and to find out relations between IP and AM with verifiers of different power. AM with  verifiers modeled by Turing Machines augmented with a fixed-size quantum register  (qAM) were studied also  in \cite{Yak12a,Yak12} and it was shown that the an NP-complete language $L_{knapsack}$, representing the 0~-~1 knapsack problem,  can be recognized by a qAM whose verifier is a {\em two-way finite automaton} working on  quantum mixed states  using superoperators. In Yakaryilmaz's notation, {\em two-way finite automata} working on  quantum mixed states  using  superoperators are called 2QCFA. However, 2QCFA as defined originally  in \cite{Amb02}, are working only on quantum pure  states using unitary operators. They  can be simulated efficiently by  {\em two-way finite automata} working on  quantum mixed states, but  whether  {\em two-way finite automata} working on  quantum mixed states can be simulated by  2QCFA, or not, is unknown. The model of 2QCFA we use is that of \cite{Amb02} and it is weaker, actually a special case of the model used in \cite{Yak12a,Yak12}. Our results concerning the acceptance of the language $L_{knapsack}$ are therefore stronger. It is also worth mentioning that a notion of QMA for quantum-automata verifiers was introduced in \cite{NY04,Nis09,NY14} (under the name ``public QIP").

Our model will be denoted as QAM(2QCFA). One can see this model also as a classical AM
augmented with a quantum memory of constant size~--~to store quantum states of a fixed
Hilbert space~--~that does not depend on the size of the (classical) input.
Our main results show that such models are more powerful than AM(2PFA)~--~that is AM with 2PFA as verifiers, and the NP-complete language can be recognized by QAM(2QCFA).

The paper is structured as follows. In Section 2 all models involved are
described in detail. After that we show for the language $L_{middle}=\{xay\,|\, x,y \in \{a,b\}^*, |x|=|y|\}$ that for any $0\leq\varepsilon<1/2$ there is a QAM(2QCFA) $A(P, V_{\varepsilon})$~--~with the prover $P$ and the verifier $V_{\varepsilon}$ that accepts $L_{middle}$ with one-sided error $\varepsilon$ in a polynomial expected running time~--~notation QAM(\mbox{ptime-2QCFA}). This language cannot be recognized by any AM(2PFA) in polynomial expected running time, as shown in \cite{DwS92}. As we will show in the paper, for the language $L_{mpal}=\{xax^R\,|\, x\in \{a,b\}^* \}$, that for any $0\leq\varepsilon<1/2$ there is a QAM(2QCFA) $A(P,V_{\varepsilon})$ that can recognize $L_{mpal}$ with one-sided error $\varepsilon$ in an exponential expected running time. We will prove that this language cannot be recognized at all by an AM(2PFA). These results show that QAM(2QCFA) are more powerful than AM(2PFA).  Afterwards we show that there is an NP-complete language, namely $L_{knapsack}$,  representing the 0~-~1 knapsack problem, that can be recognized by QAM(2QCFA) in an exponential expected running time. Finally, we discuss languages, $L_1=\{w\,|\, \exists s,t,u,v\in \{a,b\}^*,\ w=sbt=ubv,|s|=|v|\}$ and $L_2=\{w\,|\, \exists s,t,u,v\in \{a,b\}^*,\ w=sat=ubv,|s|=|v|\}$, that can be recognized by QAM(ptime-2QCFA). The language $L_1$ is proved to be nonstochastic. The language $L_2$, the set of nonpalindromes, is stochastic \cite{Fre10}.

The fact that the {\em non-regular language} $L_{middle}$ can be recognized by a QAM(ptime-2QCFA)
and it seems that it can not be recognized by a 2QCFA, indicates that QAM(ptime-2QCFA) are likely more powerful than  2QCFA. Interestingly enough, this situation seems to be different for 2PFA. It is still an open problem to find out whether there is a {\em non-regular language} that can be recognized by AM(ptime-2PFA), but we know that any 2PFA needs  exponential time to recognize a {\em non-regular language} \cite{DwS90,Gre86} .

\section{Basic models}
At first we introduce formally the model IP(2PFA) and afterwards also the model
QIP(2QCFA).
Concerning basics of quantum computation we refer the reader to \cite{Gru99,Nie00},
and concerning basics of classical and quantum automata we refer the reader to \cite{Gru99,Gru00,Hop79,Paz71,Qiu12}.

\subsection{Model IP(2PFA)}

Notation: A {\em coin-tossing distribution} on a finite set $S$ is a mapping $\phi:S\to \{0,1/2,1\}$ such that $\sum_{s\in S}\phi(s)=1$.

\begin{figure}[htbp]
  \centering\epsfig{figure=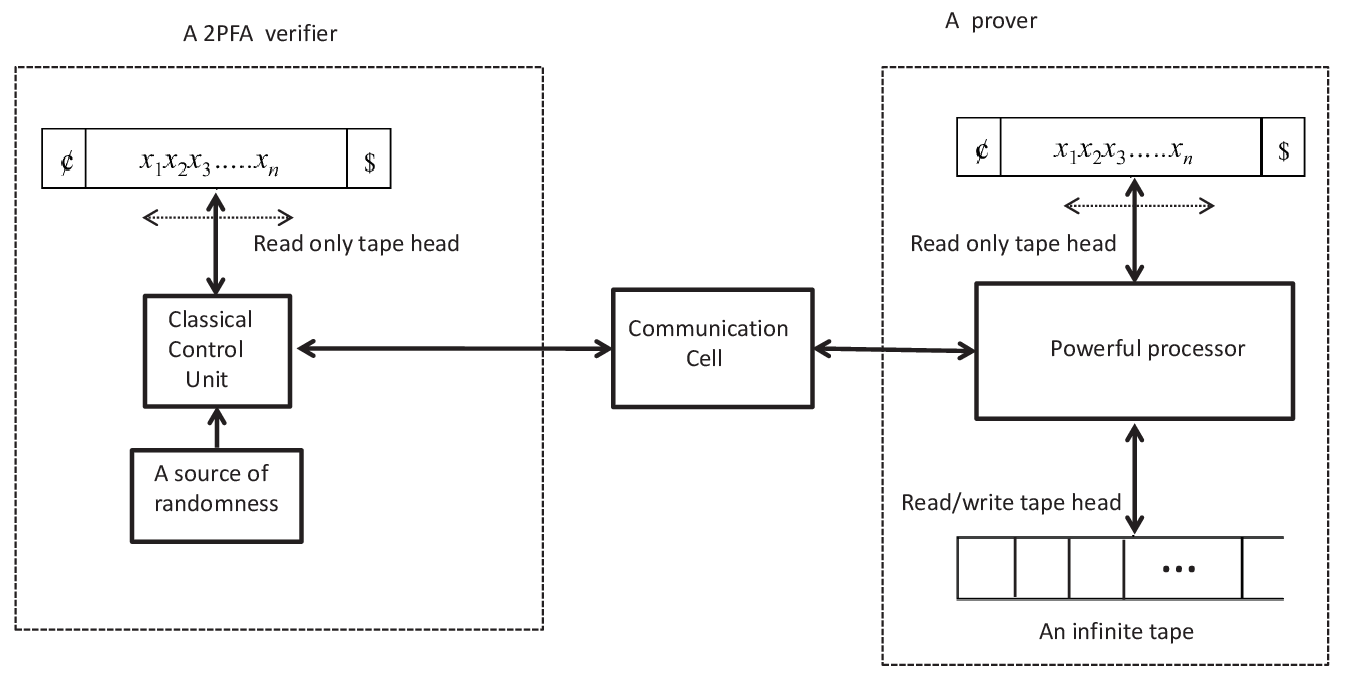,width=130mm}\\
 \centering\caption{A schematic model of an IP(2PFA)}\label{f1}
\end{figure}

\begin{Df}

An IP(2PFA) $A=(P,V)$, where $P$ is a so-called {\em prover}
and $V$ is a so called {\em verifier} (that are specified bellow and communicate through
a communication cell as illustrated in Figure \ref{f1}).

An action of $A$ starts with a step of the verifier. The verifier's head scans the left end-marker, and the verifier writes a symbol to the communication cell. That is followed by an action of the prover who writes its response into the communication cell. Such rounds of steps continue until the verifier decides to end it.

 The verifier $V$ is a 2PFA specified as follows:
$V=(S, \Sigma,\delta, s_0, S_{acc}, S_{rej})$, where

\begin{enumerate}

\item $S$ is a finite set (of classical states) partitioned into subsets $R, C, H$ of reading, communication and halting states, respectively.

\item $s_0\in R$ is the initial reading state.

\item $H$ is partitioned into sets $S_{acc}$ and $S_{rej} $ of the accepting and rejecting states.

\item $\Sigma$ is a finite input alphabet that is extended into the alphabet $\Sigma'=\Sigma\cup\{\
|\hspace{-1.5mm}c,\$\}$, where $\ |\hspace{-1.5mm}c\notin \Sigma $
will be used as the left end-marker and $\$\notin \Sigma$ will be used as the
right end-marker. $\Gamma$ is a communication alphabet (shared by both the verifier and the prover).

\item $\delta$ is a transition mapping defined as follows:
\begin{enumerate}
\item For each reading state $s\in R$ and $\sigma\in \Sigma'$, $\delta(s,\sigma)$ is a coin-tossing distribution on $S\times \{-1,0,1\}$, where $-1, 0$ and $1$ specify the tape head move~--~one cell left, no move, one cell right. It is assumed that $\delta$ is well defined in the sense that
on end-markers it never determines a move outside of the tape region separated by end-markers $
|\hspace{-1.5mm}c$ and $\$$.  Therefore $\delta$ specifies, for each state of the control unit and each symbol on the tape, the probability  that
the control unit will be in a particular new state and that the tape head moves in a particular direction.

 \item With each communication state $s\in C$, a unique communication symbol $\gamma_s$ is associated.
In the case the verifier is in such a communication state $s$, then it writes $\gamma_s$ in the communication cell and then the prover gets into an action. It reads the symbol stored in the communication cell and, depending on the whole communication history and the input string (that is the same for both the prover and the verifier), the prover writes a symbol $\gamma\in \Gamma$ into the communication cell. Afterwards the verifier gets into the action that depends on its current state and on the symbol in the communication cell. For each pair $s\in S$ and $\gamma\in \Gamma$ a coin-tossing distribution function on the set $S$ is defined that determines the probability for each state from $S$ that it should be the next state of the verifier. At such after-communication step the head of the verifier does not move.
\item The verifier halts when it gets into a halting state. Therefore
there is no need to define $\delta$ on states in $H$.

\end{enumerate}
\end{enumerate}

The prover $P$ is all-powerful and at each communication step what the prover writes into the communication cell depends on the whole input and the whole communication history up to that point of communication. Namely, it is determined, for an input string $x$ and communication history $y\in \Gamma^*$ to that point, by a coin-tossing distribution $\rho(x,y)$ defined on $\Gamma$.

Note. The communication cell always holds a symbol from the communication alphabet $\Gamma$.  Whenever the verifier needs some information from the prover, the verifier writes
a request to the communication cell via a symbol from the alphabet $\Gamma$ and the prover responds. The verifier then continues, probabilistically, depending on the prover's respond.

\end{Df}

\begin{Rm}

In the above definition, the prover's understanding of the verifier's computation is only through the communication cell that can contain information only from the finite set $\Gamma$. This mode of communication is called {\em private coins} \cite{Gol86} or {\em partial information} \cite{Con88}. One can consider also IP in which the prover has complete information on the current configuration of the verifier. Such a communication mode is called {\em public coins}  \cite{Gol86} or  {\em complete information} \cite{Con88} mode. In Babai's terminology, an IP with the public  coins communication mode is called Arthur-Merlin proof system (AM). In an AM(2PFA), the verifier in each step automatically sends its configuration information, that is the communication symbol corresponding to the communication alphabet contains an element of $S\times \{-1,0,1\}$, that is, the current state and the last move of the head, to the prover through the communication cell in every computation step.  This information is sufficient for the all-powerful prover to keep the track of the configuration of the verifier because the prover knows the strategy of the verifier.
\end{Rm}

The computation of $(P,V)$ on an input string $w$ starts with the string $|\hspace{-1.5mm}cw\$$  on the input tapes of both the verifier and the prover. The tape head of the verifier is positioned on the left end-marker $|\hspace{-1.5mm}c$ and the verifier begins to act with the initial state $s_0$ in its control unit.
The action of the verifier and the prover in the next steps is then governed (probabilistically) by the transition functions $\delta$ and $\rho$, as defined above, until the verifier enters a halting state. For a particular input $w$ and a halting state $s$ let $Pr_{(P,V)}(w,s)$ be the probability that IP $(P,V)$ halts its computation on $w$ in the state $s$. The probability that the verifier halts in a halting state $s$ on the input $w$ is taken over all random choices of the verifier and the prover. If $s$ is in an accepting (a rejecting) state, then the input is accepted (rejected).

The prover-verifier pair
$(P,V_{\varepsilon})$ is an AM(2PFA) for (accepting) a language $L\subset \Sigma^*$ with an error probability $\varepsilon<1/2$ if
\begin{enumerate}
\item {\bf (Completeness condition)}: for all $w\in L$, $Pr[(P,V_{\varepsilon})\ \mbox{ accepts}\ w]\ge 1-\varepsilon$, and
\item {\bf (Soundness condition)}: for all $w\not\in L$ and any prover $P^*$, $Pr[(P^*,V_{\varepsilon}) \ \mbox{rejects}\ w]\ge 1-\varepsilon$.
\end{enumerate}

We say that a language $L$ is recognized by AM(2PFA) if for some $\varepsilon<1/2$, there is an AM(2PFA) $(P, V_{\varepsilon})$ that accepts the language $L$ with the error probability
$\varepsilon$.

\subsection{ Model QIP(2QCFA)}

2QCFA were introduced by Ambainis and Watrous \cite{Amb02} and explored by Qiu, Yakaryilmaz and others \cite{Qiu08,Yak10,ZhgQiu11,Zhg12}. Informally, a 2QCFA can be seen as a 2DFA with an access to a quantum memory for states of a fixed Hilbert space upon which at each step either a unitary operation is performed or a projective measurement and the outcomes of which then probabilistically determines the next move of the underlying 2DFA.

A quantum analogue of IP  was introduced by Watrous in 2003 \cite{Wat03} and since then the study of QIP has become an interesting and important topic of quantum complexity theory. In particular, quantum analogues of Babai's {\em Arthur-Merlin proof system} called now {\em quantum Arthur-Merlin proof system} (QAM) have drawn a significant attention \cite{Wat09,Yak12,Yak12a}. In \cite{Yak12,Yak12a}, for Arthur-Merlin IP verifiers are augmented by a fixed-size quantum memory. In \cite{NY04,Nis09,NY14} {\em one-way quantum finite automata} (both of the {\em measure-many} and {\em measure-once} types) and {\em two-way quantum finite automata} are considered as verifiers.

In this paper, such IP are mainly considered in which verifiers are not 2PFA, as in \cite{NY04,Nis09,NY14},
but 2QCFA  and the verifier and the prover have just classical communication.
More exactly, verifiers are augmented by a quantum memory size of which does not depend on the size of the input. The formal definition is as follows:

\begin{figure}[htbp]
 \centering\epsfig{figure=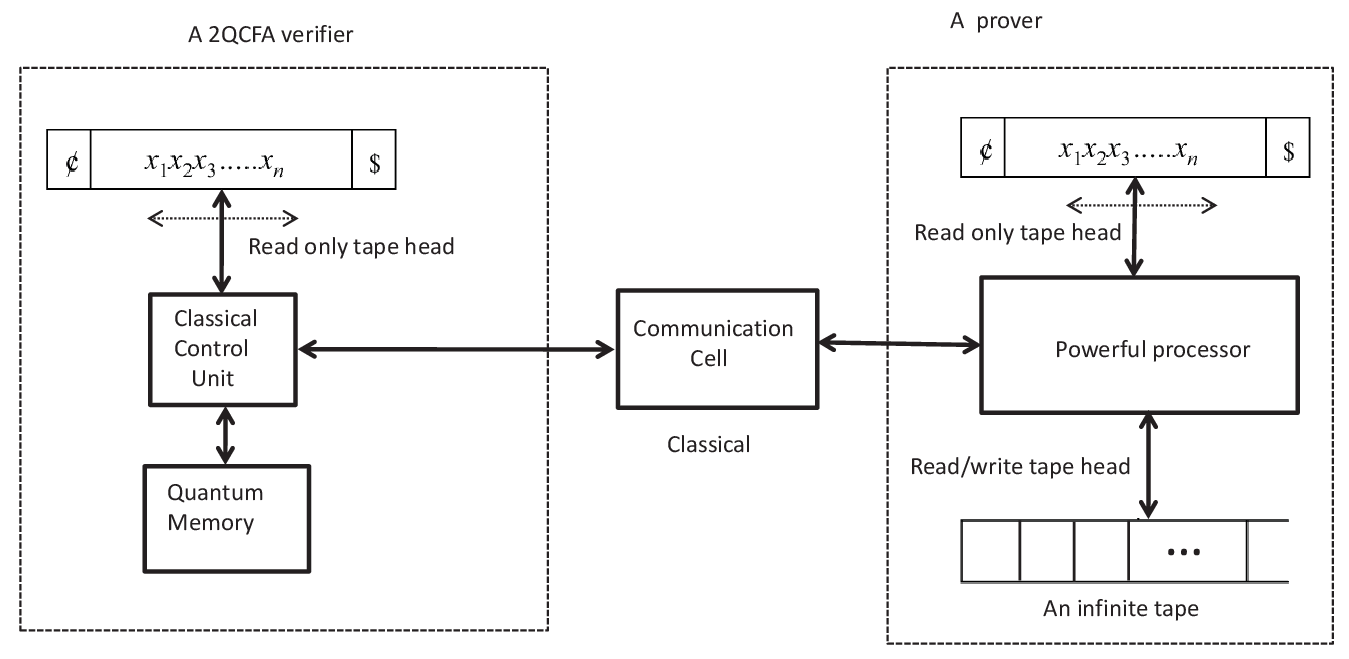,width=130mm}\\
  \centering\caption{A schematic model of a QIP(2QCFA)}\label{f2}
\end{figure}

\begin{Df}
A QIP(2QCFA) is given by a pair $(P,V)$, where $P$ is a prover and $V$ is a verifier, as illustrated in Figure \ref{f2}, and the prover communicate classically with the verifier  through a special communication cell.

In such a QIP(2QCFA) $(P,V)$, the verifier $V$ is a 2QCFA $A=(Q,S,\Sigma,\Gamma,  \Theta, \delta, |q_0\rangle, s_0, S_{acc}, S_{rej})$, where

\begin{enumerate}
\item Q is a finite set of an $n$ basic orthonormal quantum states.
\item S is a finite set of classical states that is partitioned into subsets $R, C, H$ of the reading, communication and halting states..
\item $\Sigma$ is a finite set of input symbols extended to the tape symbols
$\Sigma'=\Sigma\cup \{|\hspace{-1.5mm}c,\$\}$, where $\ |\hspace{-1.5mm}c$ is used as the left end-marker and $\$$ is used as the right end-marker.
\item $\Gamma$ is a finite set of communication symbols.
\item $|q_0\rangle\in Q$ is the initial quantum state.
\item $s_0$ is the initial classical state.
\item The set $H$ is partitioned into subsets $S_{acc}$ and $S_{rej}$ of the accepting and rejecting states.
\item $\Theta$ is a mapping
$\Theta: R\times \Sigma'\to U(H(Q))\cup O(H(Q))$,
where U(H(Q)) and O(H(Q)) are sets of unitary operations and projective measurements on the Hilbert space generated by quantum states from $Q$.
\item $\delta$ is a classical transition function defined as follows.

\begin{enumerate}
\item For each reading state $s\in R$ and a tape symbol $\sigma$
\begin{enumerate}
\item If $\Theta(s,\sigma)\in U(H(Q))$, then the unitary operation $\Theta(s,\sigma)$ is applied on the current state of quantum memory to produce  a new quantum state and, in addition, if
$$\delta: R\times \Sigma'\to S\times \{-1, 0,1\},$$
then, for the case $\delta(s,\sigma)=(s',d)$, the new classical state of the verifier is $s'$ and its head moves in the direction $d$.
\item If $\Theta(s,\gamma)\in O(H(Q))$, then $\Theta(s,\sigma)$ is a projective measurement with a set of possible eigenvalues
$E_{\Theta(s,\lambda)}=\{\lambda_1,\ldots, \lambda_n\}$ and projectors $\{P_1,\ldots, P_n\}$, where $P_i$ is the projector onto the eigenspaces generated by eigenvectors corresponding to $\lambda_i$. In such a case
$$\delta_{(s,\sigma)}: E\to S\times\{-1,0,1\}$$
and $\delta_{(s,\sigma)}(\lambda_i) =(s',d)$ means that when the projective measurement outcome is $\lambda_i$, then the new classical state is $s'$ and its head moves in the direction $d$.
\end{enumerate}

\item With each communicating state $s\in C$ a communication symbol $\gamma_s$ is associated. For each $s\in C$ and $\gamma\in \Gamma$, $\delta(s,\gamma)\in S$.
This has the following meaning:
If the verifier gets into a state $s\in C$, it writes $\gamma_s$ into the communication cell and then the prover, depending on the whole communication history and on the input $w$ writes a $\gamma\in \Gamma$ into the communicating cell and the verifier comes into the action. If $\delta(s,\gamma)=s'$, then $s'$ is the new state of the verifier and the tape head of the verifier does not move.
\item The verifier halts and accepts (rejects) the input when it enters a classical accepting (rejecting) state from $H$.
\end{enumerate}

\end{enumerate}

The prover is a processor of an unlimited power as in the general case of IP.

\end{Df}

\begin{Rm}
 A QIP(2QCFA) is a QAM(2QCFA) if the verifier sends, at the end of each its communication round through a communication symbol information about its current configuration, that is an element of  $S\times \{-1,0,1\}\times E$, where $E$ is the set of all possible measurement outcomes,  that is the communication symbol contains the current classical state, the last move of the head and the projective measurement outcome (let a special symbol is used if there was no measurement in the computation step), to the prover through the communication cell in every computation step.  This information is sufficient for the all-powerful prover to keep the track of the configuration of the verifier because the prover knows the strategy of the verifier.
\end{Rm}

The computation of QIP(2QCFA) with a prover $P$ and a verifier $V$ on an input $w\in \Sigma^*$ starts with the string $|\hspace{-1.5mm}cx\$$ on the input tapes of both of them. At the start, the tape head of the verifier is positioned on the left end-marker and the verifier begins the computation in the classical initial state and
in the initial quantum state.
In each of the next steps, if the current classical state of the verifier is $s$ and the current quantum state of the verifier is $|\psi\rangle$ and the scanning symbol is $\sigma$, then the quantum and classical states are changed according to $\Theta(s,\sigma)$ and $\delta$ as follows.
\begin{enumerate}

\item If $s\in R$, then:

\begin{enumerate}
\item If $\Theta(s,\sigma)$ is a unitary operation $U$, then $U$ is applied on the current quantum state $|\psi\rangle$, changing it into the quantum state $U|\psi\rangle$, and
$\delta(s,\sigma)=(s',d)$ makes $s'$ to be a new classical state and the head moves in the direction $d$. If $s\in S_{acc}$ ($s\in S_{rej}$),
then the input is accepted (rejected).
\item If $\Theta(s,\sigma)$ is a projective measurement, then the current quantum state is changed to $P_j|\psi\rangle/||P_j|\psi\rangle||$ with the probability $||P_j|\psi\rangle||^2$ and in such a case $\delta_{(s,\sigma)}$ is a mapping from the set of potential classical outcomes (eigenvalues) of the measurement to $S\times \{-1,0,1\}$.  In particular, if the measurement coutcome is $\lambda_i$ and
     $\delta_{(s,\sigma)}(\lambda_i)=(s',d)$, then:

\begin{enumerate}
\item if $s'\in S-H$, then $s'$ is the new classical state and the head moves in the direction $d$;
\item if $s'\in S_{acc}$ ($s'\in S_{rej}$), then the verifier accepts (rejects) the input and computation halts.
\end{enumerate}
\end{enumerate}

\item If the current classical state  $s\in C$, then the verifier sends $\gamma_s$ to the prover through the communication cell. Depending on this and all previously obtained symbols from the verifier, as well as on the input string, the prover sends to the communication cell a symbol $\gamma$.  If
$\delta(s,\gamma)=s'$, then $s'$ will become the new classical state of the verifier and the input head does not move. In case $s'\in S_{acc}$ ($s'\in S_{rej}$) the input is accepted.
\end{enumerate}

The probability that an input word is accepted is defined in a similar way as in the case of 2QCFA.

The prover-verifier pair $(P,V_{\varepsilon})$ is a QAM(2QCFA) for (accepting) a language $L\subset\Sigma^*$ with one-sided error $0\leq\varepsilon<1/2$ if
\begin{enumerate}
\item {\bf Completeness condition}: for all $w\in L$, $Pr[(P,V_{\varepsilon})\ \mbox{ accepts}\ w]=1$, and
\item {\bf Soundness condition}: for all $w\not\in L$ and any prover $P^*$, $Pr[(P^*,V_{\varepsilon}) \ \mbox{rejects}\ w]\ge 1-\varepsilon$.
\end{enumerate}

We say that a language $L$ is recognized by QAM(2QCFA) if for some $\varepsilon<1/2$ there is a QAM(2QCFA) $(P,V_{\varepsilon})$ that accepts the language $L$ with one-sided error $\varepsilon$.

\section{Examples  of languages recognized by QAM(2QCFA)}
In this section we provide a detailed proof for five languages that they are
accepted by QAM(2QCFA).

\subsection{Recognition of the language $L_{middle}$}
The importance of the fact that the language $L_{middle}$ can be recognized by a QAM(ptime-2QCFA) is underlined by the fact that there is no AM(ptime-2PFA) for this language \cite{DwS92}.

\begin{Th}\label{L-middle}

For any $\varepsilon<1/2$ there exists a QAM(2QCFA) {\em $A_{\varepsilon}$} with a
verifier-prover pair $(P, V_{\varepsilon})$ for the language
$L_{middle}$ with one-sided error $\varepsilon$ in the expected running time
$\textbf{O}\left(\frac{1}{\varepsilon}\left(n^4+n^2\log{\frac{1}{\varepsilon}}\right)\right)$, where $n$ is the length of the input.
\end{Th}

\begin{proof}

\begin{figure}[htbp]
\centering\epsfig{figure=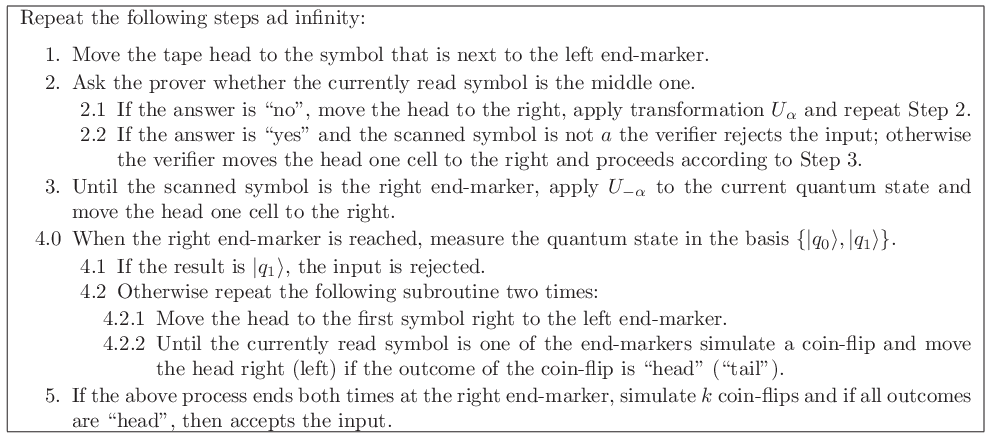,width=140mm}\\
 \centering\caption{ Description of the behavior of the verifier $V_{\varepsilon}$ in a QAM recognizing the language $L_{middle}$. The choice of $k$  depends on $\epsilon$, as discussed in the text. }\label{f3}
\end{figure}

At first we present informally the main idea of the proof. A 2QCFA
verifier $V_{\varepsilon}$ uses quantum memory with states generated by two orthogonal quantum states $|q_0\rangle$ and $|q_1\rangle$, where $|q_0\rangle$ will be the initial quantum state.
At the  beginning the input is shared by both the verifier and the prover.
In the first part of the interaction/computation process, until the middle symbol is reached, for each input symbol $\sigma$ the verifier does the following.

It applies the following unitary transformation $U_{\alpha}$, $\alpha=\sqrt{2}\pi$,
\begin{equation}
U_{\alpha}|q_0\rangle=\cos \alpha |q_0\rangle+\sin \alpha|q_1\rangle,\ \ \ U_{\alpha}|q_1\rangle=-\sin \alpha |q_0\rangle+\cos \alpha|q_1\rangle,
\end{equation}
 to the current quantum state (that is the state is rotated by the angle $\alpha$)
and asks the prover whether the next symbol is the middle one. The all-powerful
prover provides the answer. If the next symbol is not the middle one, the verifier moves its head one cell to the right and this process repeats. When the symbol in the middle is reached, the verifier checks whether it is the symbol $a$. If not, the verifier rejects the input; otherwise it continues to read the input, symbol by symbol, and each time it
applies the unitary transformation $U_{-\alpha}$ to the current quantum state (and therefore this quantum state is rotated by the angle $-\alpha$) until the right end-marker is reached. (During this process  the verifier has no need to get any information from the prover.)

When the right end-marker is reached, the verifier measures the current quantum state in the basis $\{|q_0\rangle, |q_1\rangle\}$. If $|q_1\rangle$ is the resulting quantum state, the input is rejected, otherwise the verifier proceeds as shown in Figure \ref{f3}.

\begin{Lm}

If the input $w\in L_{middle}$, then the quantum states of the verifier $V_{\varepsilon}$  evolve with certainty into $|q_0\rangle$ when the right end-marker is reached in Step 4.
\end{Lm}
\begin{proof}
If $w\in L_{middle}$, then there are strings $x,y\in\Sigma^*$ such that $|x|=|y|$ and $w=xay$. Since the all-powerful prover can tell the verifier for sure when the middle symbol is reached, the quantum state when the right end-marker is reached is
\begin{equation}
|q\rangle=(U_{-\alpha})^{|y|} (U_{\alpha})^{|x|}
|q_0\rangle=\left(
  \begin{array}{cc}
    \cos \alpha  & -\sin \alpha \\
    \sin \alpha  & \cos  \alpha \\
  \end{array}
\right)^{|y|}\left(
  \begin{array}{cc}
    \cos \alpha  & \sin \alpha \\
    -\sin\alpha  & \cos \alpha \\
  \end{array}
\right)^{|x|}|q_0\rangle
\end{equation}
\begin{equation}
=\left(
  \begin{array}{cc}
    \cos |y|\alpha  & -\sin |y|\alpha \\
    \sin |y|\alpha  & \cos  |y|\alpha \\
  \end{array}
\right)\left(
  \begin{array}{cc}
    \cos |x|\alpha  & \sin |x|\alpha \\
    -\sin |x|\alpha  & \cos |x|\alpha \\
  \end{array}
\right)|q_0\rangle=\left(
  \begin{array}{cc}
    1  & 0 \\
    0  & 1 \\
  \end{array}
\right)|q_0\rangle=|q_0\rangle.
\end{equation}

\end{proof}

\begin{Lm}\cite{Zhg12}
 A coin flipping can be simulated by the verifier $V_{\varepsilon}$ using states
$|q_0\rangle$ and $|q_1\rangle$.
\end{Lm}

\begin{Lm}\cite{Amb02}
 If the input $w\in L_{middle}$, then the execution of loops in the steps 4.2 and 5 leads to the acceptance with the probability $\frac{1}{2^k(n+1)^2}$.
\end{Lm}

It may happen that the process described in Figure \ref{f3} does not terminate. This happens only if the result of the last coin-flip is ``tail" and this happens either when the left end-marker
is reached in Step 4.2 or this happens in Step 5. In such a case the quantum state of the verifier is $|q_0\rangle$ and a new iteration of the process starts.

{\bf Completeness condition:} If $w\in L_{middle}$, then the quantum state of the verifier after Step 4.0 is $|q_0\rangle$ and the input is never rejected in Step 4.1. After the rest of the steps,
in 4.2 and 5, the probability of accepting the input is, if we denote $n=|w|$,
$P_a=1/(2^k(n+1)^2)$ according to the previous lemma, and the probability of rejecting the input is $P_r=0$. If  the whole process is repeated for infinitum, the acceptance probability is
\begin{align}
Pr[(P^*,V_{\varepsilon})\  \text{ rejects}\  w] =\sum_{i\geq
0}(1-P_a)^i(1-P_r)^iP_a
=\frac{P_a}{P_a+P_r-P_aP_r}
=\frac{P_a}{P_a}=1.
\end{align}

{\bf Soundness condition}: Let the input string $w\not \in L_{middle}$ and $n =|w|$. Observe
that the verifier, in its communication with the prover, waits only for the
information that the currently read symbol is in the middle and let this come after reading $m
(0<m<n)$ symbols in the input string (this $m$ can be different at different iterations and for different provers).

\begin{Lm}\label{not-in-L-middle}
If $w\not \in L_{middle}$, then the verifier rejects the input for any prover after Step 4.1  with probability at least $1/(2n^2+1)$.
\end{Lm}

\begin{proof}
Suppose that a prover $P^*$,  after $m$ steps, informs the verifier that the next symbol is the middle one. If $w$ does not have the form $xay$, $|x|=m$, then $w$ is rejected in Step 2.2.

If $w=xay$ with $|x|=m$, then, before the measurement in Step 4.0 the quantum state of the verifier will be
\begin{equation}
|q\rangle=(U_{-\alpha})^{n-m-1} (U_{\alpha})^m|q_0\rangle=\left(
  \begin{array}{cc}
    \cos \alpha  & -\sin \alpha \\
    \sin\alpha  & \cos \alpha \\
  \end{array}
\right)^{n-m-1}\left(
  \begin{array}{cc}
    \cos \alpha  & \sin \alpha \\
    -\sin \alpha  & \cos  \alpha \\
  \end{array}
\right)^m|q_0\rangle
\end{equation}
\begin{equation}
=\left(
  \begin{array}{cc}
    \cos (n-m-1)\alpha  & -\sin (n-m-1)\alpha \\
    \sin (n-m-1)\alpha  & \cos (n-m-1)\alpha \\
  \end{array}
\right)\left(
  \begin{array}{cc}
    \cos m\alpha  & \sin m\alpha \\
    -\sin m\alpha  & \cos  m\alpha \\
  \end{array}
\right)|q_0\rangle
\end{equation}
\begin{equation}
=\left(
  \begin{array}{cc}
    \cos (n-2m-1)\alpha  &\sin (n-2m-1)\alpha \\
    \sin (n-2m-1)\alpha  & \cos  (n-2m-1)\alpha \\
  \end{array}
\right)|q_0\rangle
\end{equation}
\begin{equation}
=\cos\left((n-2m-1)\alpha\right)|q_0\rangle+\sin\left((n-2m-1)\alpha\right)|q_1\rangle.
\end{equation}

The probability of observing  $|q_1\rangle$ is
$\sin^2\left(\sqrt{2}(n-2m-1)\pi\right)$ in Step 4.0.
Without a loss of generality, we assume that $n-2m-1>0$. Let $l$ be
the closest integer to $\sqrt{2}(n-2m-1)$. If
$\sqrt{2}(n-2m-1)>l$, then $2(n-2m-1)^2>l^2$. So we get
$2(n-2m-1)^2-1\geq l^2$ and $l\leq \sqrt{2(n-2m-1)^2-1}$.
Therefore
\begin{equation}
\sqrt{2}(n-2m-1)-l\geq \sqrt{2}(n-2m-1)-\sqrt{2(n-2m-1)^2-1}
\end{equation}

\begin{equation}
=\frac{(\sqrt{2}(n-2m-1)-\sqrt{2(n-2m-1)^2-1})(\sqrt{2}(n-2m-1)+\sqrt{2(n-2m-1)^2-1})}{\sqrt{2}(n-2m-1)+\sqrt{2(n-2m-1)^2-1}}
\end{equation}
\begin{equation}
=\frac{1}{\sqrt{2}(n-2m-1)+\sqrt{2(n-2m-1)^2-1}}>\frac{1}{2\sqrt{2}(n-2m-1)}.
\end{equation}

Because $l$ is the closest integer to $\sqrt{2}(n-2m-1)$, we have
$0<\sqrt{2}(n-2m-1)-l<1/2$. Let $f(x)=sin(x\pi)-2x$. We have
 $f''(x)=-\pi^2\sin(x\pi)\leq 0$ when $x\in
[0,1/2]$. That is to say, $f(x)$ is concave in  the interval $[0,1/2]$, and we
have $f(0)=f(1/2)=0$. So, for any $x\in[0,1/2]$,
$f(x)\geq 0$, that is $\sin(x\pi)\geq 2x$. Therefore, we have
\begin{equation}
\sin^2(\sqrt{2}(n-2m-1)\pi)=\sin^2\left((\sqrt{2}(n-2m-1)-l)\pi\right)
\end{equation}
\begin{equation}
\geq \left(2(\sqrt{2}(n-2m-1)-l)\right)^2
=4\left(\sqrt{2}(n-2m-1)-l\right)^2
\end{equation}
\begin{equation}
> 4\left(\frac{1}{2\sqrt{2}(n-2m-1)}\right)^2=\frac{1}{2(n-2m-1)^2}>\frac{1}{2(n-2m-1)^2+1}.
\end{equation}

If $\sqrt{2}(n-2m-1)<l$, then $2(n-2m-1)^2<l^2$. So we get
$2(n-2m-1)^2+1\leq l^2$ and $l\geq \sqrt{2(n-2m-1)^2+1}$. We have therefore
\begin{equation}
\sqrt{2}(n-2m-1)-l\leq \sqrt{2}(n-2m-1)-\sqrt{2(n-2m-1)^2+1}
\end{equation}
\begin{equation}
=\frac{\left(\sqrt{2}(n-2m-1)-\sqrt{2(n-2m-1)^2+1}\right)\left(\sqrt{2}(n-2m-1)+\sqrt{2(n-2m-1)^2+1}\right)}{\sqrt{2}(n-2m-1)+\sqrt{2(n-2m-1)^2+1}}
\end{equation}
\begin{equation}
=\frac{-1}{\sqrt{2}(n-2m-1)+\sqrt{2(n-2m-1)^2+1}}<\frac{-1}{2\sqrt{2(n-2m-1)^2+1}}
\end{equation}
and this implies that
\begin{equation}
l-\sqrt{2}(n-2m-1)>\frac{1}{2\sqrt{2(n-2m-1)^2+1}}.
\end{equation}
Because $l$ is the closest integer to $\sqrt{2}(n-2m-1)$, we have
$0<l-\sqrt{2}(n-2m-1)<1/2$. Therefore,
\begin{equation}
\sin^2\left(\sqrt{2}(n-2m-1)\pi\right)=\sin^2\left((\sqrt{2}(n-2m-1)-l)\pi\right)
\end{equation}
\begin{equation}
=\sin^2\left((l-\sqrt{2}(n-2m-1))\pi\right)\geq
\left(2(l-\sqrt{2}(n-2m-1))\right)^2
\end{equation}
\begin{equation}
=4\left(l-\sqrt{2}(n-2m-1)\right)^2>
4\left(\frac{1}{2\sqrt{2(n-2m-1)^2+1}}\right)^2=\frac{1}{2(n-2m-1)^2+1}.
\end{equation}

As  $0<m<n$, we have $|n-2m-1|<n$, and therefore
\begin{equation}
\frac{1}{2(n-2m-1)^2+1}>\frac{1}{2n^2+1}.
\end{equation}
So the lemma has been proved.
\end{proof}

If $w\not\in L_{middle}$, then for any prover $P^*$ the above verifier rejects the input after  Step 4.1 with the probability
\begin{equation}
P_{r} >\frac{1}{2n^2+1}
\end{equation}
according to Lemma \ref{not-in-L-middle}. Moreover, after the last two steps the verifier accepts $w$ with the probability
\begin{equation}
P_{a}=\frac{1}{2^k(n+1)^2}.
\end{equation}

If $k=1+\lceil\log_2 1/\varepsilon\rceil$, then $\varepsilon\geq
1/2^{k-1}$.

If the whole process is repeated indefinitely, then the probability that the verifier rejects the input $w$ for any prover $P^*$  is

\begin{align}
Pr[(P^*,V_{\varepsilon})\  rejects\  w] =&\sum_{i\geq
0}(1-P_a)^i(1-P_r)^iP_r
=\frac{P_r}{P_a+P_r-P_aP_r}>\frac{P_r}{P_a+P_r}\\
>&\frac{1/(2n^2+1)}{\varepsilon/2(n+1)^2+1/(2n^2+1)}
=\frac{(n+1)^2/(2n^2+1)}{\varepsilon/2+(n+1)^2/(2n^2+1)}
\end{align}

If we now denote $f(x)=\frac{x}{\varepsilon/2+x}=1-\frac{\varepsilon}{(\varepsilon+2x)}$, then $f(x)$ is monotonously increasing in $(0, +\infty)$ and since $(n+1)^2/(2n^2+1)\ge1/2$, we have
\begin{equation}
Pr[(P^*,V_{\varepsilon})\  rejects\
w]>\frac{1/2}{1/2+\varepsilon/2}=\frac{1}{1+\varepsilon}>1-\varepsilon.
\end{equation}

If $|w|=n$, then Steps  1 to 4.1  takes $\textbf{O}(n)$
time, the loops 4.2 takes $\textbf{O}(n^2)$ time,
and Step 5 takes $\textbf{O}(k)$ time.
The expected number of the repetitions of the algorithm
 is $\textbf{O}(2^kn^2)$ in both cases. Hence, the expected running
time of $(P,V_{\varepsilon})$ is
$\textbf{O}\left(\frac{1}{\varepsilon}\left(n^4+n^2\log{\frac{1}{\varepsilon}}\right)\right)$.
\end{proof}

\begin{Rm}
Concerning the expected running time, if the algorithm halts with the probability $h(n):N\rightarrow [0,1]$ in one iteration,
then the expected number of repetitions till the algorithm halts is
\begin{equation}\label{expected}
    \sum_{i=1}^{+\infty} (1-h(n))^{i-1}h(n)\times i=\frac{1}{h(n)}.
\end{equation}
The halting probability in one iteration in Figure \ref{f3} is ${\bf \Omega}\left(\frac{1}{2^kn^2}\right)$, so according to Equality \ref{expected}, the expected numbers of iterations of the algorithm
 is $\textbf{O}(2^kn^2)$.
\end{Rm}

\begin{Th}
Interactive proof systems QAM(ptime-2QCFA) are more powerful than
AM(ptime-2PFA).
\end{Th}
\begin{proof}
The fact that a coin-flipping can be simulated perfectly by 2QCFA in polynomial time, see \cite{ZhgQiu11}, implies that 2PFA can be simulated in polynomial time by 2QCFA and therefore QAM(ptime-2QCFA) are at least as powerful as AM(ptime-2PFA). By \cite{DwS92}, the language $L_{middle}$ cannot be recognized by AM(ptime-2PFA) and this implies that QAM(ptime-2QCFA) are more powerful than AM(ptime-2PFA).
\end{proof}

\subsection{ Recognition of the language $L_{mpal}$}

We show now that the language $L_{mpal}=\{xax^R\,|\, x\in \{a,b\}^*\}$, which cannot be recognized at all by AM(2PFA), as shown later, can be recognized by QAM(2QCFA) in an exponential expected running time.

\begin{Th}\label{L-mpal}

For any $\varepsilon$ there is a QAM(2QCFA) $A_{\varepsilon}$ with a verifier-prover pair $(P,V_{\varepsilon})$ accepting the language $L_{mpal}$ with one sided error $\varepsilon$ in the exponential running time ${\bf O}\left(n\log\frac{1}{\varepsilon}\cdot 2^{n\log{\frac{1}{\varepsilon}}}\right)$, where $n=|w|$ is the length of the input $w$.

\end{Th}
\begin{proof}

 In the QAM described in the following, the prover is used only to determine, for the verifier,  the middle symbol of the input.

\begin{figure}[htbp]
\centering\epsfig{figure=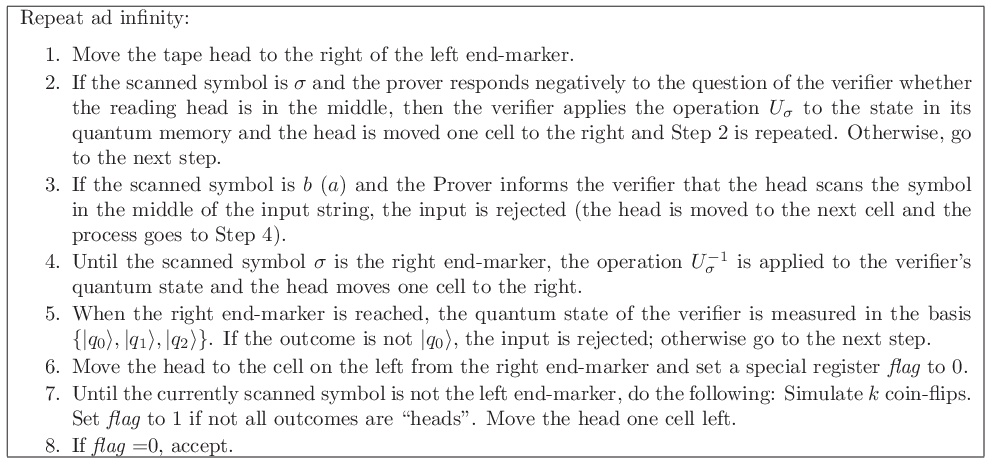,width=140mm}\\
 \centering\caption{Description of the behavior of the verifier $V_{\varepsilon}$ of the QAM $(P,V_{\varepsilon})$ when the language $L_{mpal}$ is recognized. The choice of $k$ depends on $\varepsilon$, as discussed in the text.}\label{f4}
\end{figure}

 At the beginning of the actions of the QAM, both the verifier and the prover share an input string $w$.
 The verifier $V_{\varepsilon}$ starts with the head at the left end-marker, in the initial
classical state $s_0$ and in the initial quantum  state $|q_0\rangle$ and keeps moving its head, cell by cell, till the right end-marker is reached, following the rules described informally in Figure \ref{f4}, where
\begin{equation}\label{matrix}
U_a=\frac{1}{5}\left(
  \begin{array}{ccc}
      4 & 3 & 0  \\
     -3 & 4 & 0  \\
      0 & 0 & 5
  \end{array}
\right), U_b=\frac{1}{5}\left(
  \begin{array}{ccc}
      4 & 0 & 3 \\
      0 & 5 & 0  \\
      -3& 0 & 4 \\
  \end{array}
\right)
\end{equation}
are unitary matrices that will be applied on verifier's quantum states in the Hilbert space spanned by three orthonormal basis states $|q_0\rangle, |q_1\rangle$ and $|q_2\rangle$.

An application of the unitary matrices on these basis states can also be described as follows

\begin{center}
\begin{tabular}{l|l}
  $U_{a}|q_0\rangle=\frac{4}{5}|q_0\rangle-\frac{3}{5}|q_1\rangle$ & $U_{b}|q_0\rangle=\frac{4}{5}|q_0\rangle-\frac{3}{5}|q_2\rangle$  \\
  $U_{a}|q_1\rangle=\frac{3}{5}|q_0\rangle+\frac{4}{5}|q_1\rangle$ & $U_{b}|q_1\rangle=|q_1\rangle$                                    \\
  $U_{a}|q_2\rangle=|q_2\rangle$                                   & $U_{b}|q_2\rangle=\frac{3}{5}|q_0\rangle+\frac{4}{5}|q_2\rangle$  \\
\end{tabular}
\end{center}
Let us now describe more formally the actions of the verifier $V_{\varepsilon}$
for an input $w$ with $|w|=n$.

{\bf Completeness condition}: After the start, in each round, the verifier $V_{\varepsilon}$ keeps
asking the prover whether the symbol currently scanned is in the middle of the input, until he gets a positive reply. After that the verifier has no need to get more information from the prover.

\begin{Lm}
If the input $w=xax^R$, then the quantum state of the verifier $V_{\varepsilon}$ after Step 4
is the quantum state $|q_0\rangle$.

\end{Lm}
\begin{proof}
Let $x=x_1x_2\ldots x_l$ for all $x_i \in \{a,b\}$. After
Step 4 the quantum state will be
\begin{equation}
 |\psi\rangle=U_{x_1}^{-1}U_{x_2}^{-1}\cdots U_{x_l}^{-1}U_{x_l}\cdots U_{x_2}U_{x_1}|q_0\rangle=|q_0\rangle.
\end{equation}
\end{proof}

\begin{Lm}
An execution of Steps 6 to 8 leads to an acceptance with probability $2^{-kn}$.
\end{Lm}
\begin{proof}
The probability that at the end of Step 8 flag = 0 is $2^{-kn}$. This is therefore the probability that the input is accepted in Step 8. If this is not the case,
a new iteration starts in the state $|q_0\rangle$.
\end{proof}

If the input $w\in L_{mpal}$, then the verifier never rejects the input in Step 5. After Steps 6 to 8 the input is accepted with the probability $P_a=2^{-kn}$ and is rejected with the probability $P_r=0$. If the whole computation process is repeated indefinitely, the accepting probability is
\begin{align}
Pr[(P^*,V_{\varepsilon})\  \text{ rejects}\  w] =\sum_{i\geq
0}(1-P_a)^i(1-P_r)^iP_a
=\frac{P_a}{P_a+P_r-P_aP_r}
=\frac{P_a}{P_a}=1.
\end{align}

{\bf Soundness condition}:  For a three dimensional vector $u$, let $u[i]$ denote its $i$-th component.

\begin{Lm}\label{XY}\cite{Zhg12}
Let
$$u=Y_1^{-1}\ldots Y_l^{-1}X_m\ldots X_1(1,0,0)^T$$
where $X_j, Y_j\in \{A, B\}$. If $X_j=Y_j$ for all $1\le j\le m$  and $l=m$, then $u[2]^2+u[3]^2=0$; otherwise $u[2]^2+u[3]^2>5^{-(m+l)}$.
\end{Lm}

\begin{Lm}\label{not-in-L-pal}
If the input string $w\not\in L_{mpal}$, then for any prover $P^*$ the verifier $V_{\varepsilon}$ rejects $w$  after Step 5 with the probability at least $5^{1-n}$.
\end{Lm}

\begin{proof}

Suppose that  Prover $P^*$ informs the verifier after scanning the $(m+1)$-st symbol that it is the symbol in the middle of the input string.

If the input string $w$ is not of the form $xay$ with $|x|=m$, then the verifier
rejects the input in Step 3. Assume now that $w=xay$, $|x|=m, |y|=l$ and $y\neq x^R$.

Let $x=x_1x_2\cdots x_m$ and $y=y_1y_2\cdots y_l$. Starting with the
state $|q_0\rangle$, the verifier $V_{\varepsilon}$ changes its quantum state
after the loop 4 to:
\begin{equation}
 |\psi\rangle=U_{y_l}^{-1}\cdots U_{y_2}^{-1} U_{y_1}^{-1}U_{x_m}\cdots U_{x_2}U_{x_1}|q_0\rangle.
\end{equation}
Let
$|\psi\rangle=\beta_0|q_0\rangle+\beta_1|q_1\rangle+\beta_2|q_2\rangle$.
According to Lemma \ref{XY}, $\beta_1^2+\beta_2^2>5^{-(m+l)}$. Therefore, if in Step 5, the quantum state $|\psi\rangle$ is measured, then the input is rejected with the probability
$P_r=\beta_1^2+\beta_2^2>5^{-(m+l)}=5^{1-n}$.
\end{proof}

If $w\not\in L_{mpal}$, then, according to Lemma \ref{not-in-L-pal}, for any Prover $P^*$ the verifier $V_{\varepsilon}$ rejects the input after  Step 5 with the probability
\begin{equation}
P_{r} >5^{1-n}
\end{equation}
and the input is accepted after Steps 6 to 8 with the probability
\begin{equation}
P_{a}=2^{-kn}.
\end{equation}
If the whole computation process is repeated indefinitely, the probability
that the verifier rejects the input is (taking into the consideration that $k\ge \max\{\log 5, \log\frac{1}{\varepsilon}$\})

\begin{align}
Pr[(P^*,V_{\varepsilon})\  \text{ rejects}\  w] =&\sum_{i\geq
0}(1-P_a)^i(1-P_r)^iP_r
=\frac{P_r}{P_a+P_r-P_aP_r}\\
>&\frac{P_r}{P_a+P_r}
>\frac{5^{1-n}}{2^{-kn}+5^{1-n}}
>\frac{1}{1+\varepsilon}>1-\varepsilon.
\end{align}

{\bf Time analysis:} Steps 1 to 5 take {\bf O}$(n)$ time and Steps 6 to 8
{\bf O}$(kn)$ time. The halting probability is ${\bf \Omega}\left(\frac{1}{2^{kn}}\right)$  in both cases, so the expected number of repetitions of the above process is {\bf O}($2^{kn}$) in both cases.
Hence the expected running time of the QAM $(P,V_{\varepsilon})$ is
{\bf O}$(kn\cdot 2^{kn})$ and therefore {\bf O}$\left(\log\frac{1}{\varepsilon}\cdot n\cdot 2^{n\log{\frac{1}{\varepsilon}}}\right)$.

\end{proof}

In order to prove that the language $L_{mpal}$ cannot be recognized by AM(2PFA) the following Lemma will be used.

\begin{Lm}\cite{DwS92}\label{not-in-AM(2PFA)}
Let a language $L\subseteq \Sigma^*$. Suppose there is an infinite
set $I$ of positive integers such that, for each $m\in I$, there are an integer
$N(m)$ and multisets $W_m=\{w_1,w_2,\cdots, w_{N(m)}\}$,
$U_m=\{u_1,u_2,\cdots,u_{N(m)}\}$ and
$V_m=\{v_1,v_2,\cdots,v_{N(m)}\}$ of words such that
\begin{enumerate}
    \item [(1)] $|w|\leq m$ for all $w\in W_m$,
    \item [(2)] for every integer $k$ there is an $m_k$ such that
    $N(m)\geq m^k$ for all $m\in I$ such that $m\geq m_k$, and
    \item [(3)] for all $1\leq i, j\leq N(m)$, $u_jw_iv_j\in L$
    iff $i=j$,
\end{enumerate}
then $L$ cannot be recognized by AM(2PFA).
\end{Lm}

We are now ready to prove the following theorem.
\begin{Th}\label{AM(2PFA)}
$L_{mpal}$ cannot be recognized by AM(2PFA).
\end{Th}
\begin{proof}
Let $I$ be the set of all positive integers. For each $m\in I$, let $N(m)=2^m$, and let $w_1,...,w_{N(m)}$ be an ordering of all words in $\{a,b\}^m$. Let $W_m=\{w_1,w_2,\cdots, w_{N(m)}\}$. By means of  $W_m$, $u_i=\lambda$ (the empty word) and  $v_i=aw_i^R$ for all $i$, then for all $1\le i,j\le N(m)$, $u_jw_iv_j \in L_{mpal}$ if and only if $i=j$. According to Lemma \ref{not-in-AM(2PFA)}, the language $L_{mpal}$ cannot be  recognized by AM(2PFA) and theorem has been proved.
\end{proof}

From the last theorem and from Theorem \ref{L-mpal} it follows:
\begin{Co}
QAM(2QCFA) are more powerful than AM(2PFA).
\end{Co}

\subsection{Recognition of the language $L_{knapsack}$}

 In this subsection, we consider a language, over the alphabet $\{0,1,\#\}$,
$L_{knapsack}=\{b\#a_1\#a_2\#\cdots\#a_n\,|\ \linebreak[0] \ \text{such that}\ a_1,\cdots,a_n,b\in 1\{0,1\}^*,\ \text { and there exists a set}\ I\subseteq \{1,\cdots,N\} \ \text{such that}\ v(b)=\sum_{i\in I} v(a_i)\}$, where $v(x)$ is the number such that $x$ is its binary representation.
 $L_{knapsack}$ is actually the 0~-~1 knapsack problem, which is NP-complete. Yakaryilmaz studied QAM with the verifier augmented with a fixed-size quantum register in the Arthur-Merlin
proof system \cite{Yak12,Yak12a} and proved that $L_{knapsack}$ can be recognized by QAM whose verifier is a 2QCFA which uses superoperators. Using an idea from \cite{Yak12a,Yak12} (coding of  binary strings into the amplitudes of quantum states), we  prove that the language $L_{knapsack}$ can also be recognized  by QAM(2QCFA). Our model of 2QCFA is weaker than that of  \cite{Yak12,Yak12a} and different tools, comparing those from \cite{Yak12,Yak12a},  are needed to prove our result.

\begin{Th}\label{L-knapsack}
For any $\epsilon<1/2$, there exists a verifier-prover pair
$(P,V_{\epsilon})$ of a QAM(2QCFA) to recognize the language $L_{knapsack}$ with one-sided
error $\epsilon$ in the exponential running time {\bf O}$\left(  \frac{1}{\varepsilon}(n 6^n+\log\frac{1}{\varepsilon})\right)$, where $n$ is the length of the input.
\end{Th}
\begin{proof}
Let us assume that the input string $w$ be of the form $b\#a_1\#a_2\#\cdots\#a_n$, where $b,a_i\in 1\{0,1\}^*$. This can be easily checked by an FA. Otherwise, the input string $w$ is rejected immediately. The main idea of the proof is as follows: we consider a 2QCFA verifier
$V_{\epsilon}$ that uses a quantum memory
with states generated by the set of orthogonal quantum states $\{|q_i\rangle:i=0,1,\cdots,7\}$, where $|q_0\rangle,|q_1\rangle,|q_2\rangle$ are used to encode the value of the binary strings. The verifier $V_{\epsilon}$ starts to work with the initial quantum state
$|q_0\rangle$. The tape head of the verifier
$V_{\epsilon}$ moves from the left to right. Firstly, the verifier $V_{\epsilon}$ encodes the value of $b$ into the amplitudes of the quantum state $|q_1\rangle$.
With the help of the prover $P$, the verifier $V_{\epsilon}$ will know whether $a_i$ is selected or not\footnote{Consider the equation $v(b)=\sum_{i\in I}v(a_i)=\sum_{i=1}^nc_iv(a_i)$, where $c_i\in\{0,1\}$. $c_i=1$ means $a_i$ is  selected.}. If $a_i$ is selected, the value of $a_i$ is calculated and encoded into the amplitudes of the quantum state $|q_2\rangle$ and then subtracted from the amplitude of the quantum state $|q_1\rangle$. When the right end-marker $\$$ is reached, the verifier $V_{\epsilon}$ measures
the current quantum state. If the resulting quantum state is
$|q_1\rangle$, the input string $w$ is rejected.
 Otherwise, the verifier continues as shown in Figure \ref{f5}, where the unitary operators and projective measurements are as  follows:
 \begin{figure}[htbp]
  \centering\includegraphics[width=130mm]{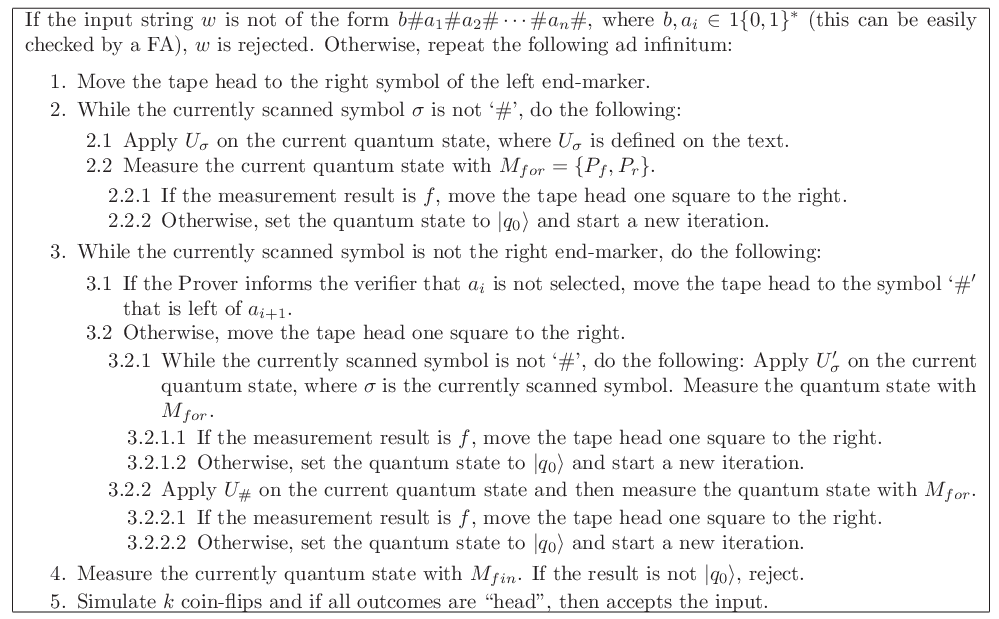}\\
  \centering\caption{Description of the behavior of the verifier $V_{\epsilon}$ of the pair $(P,V_{\epsilon})$ when recognizing the language $L_{knapsack}$. The choice of $k$ depends on $\epsilon$, as discussed in the text.}\label{f5}
\end{figure}
\begin{equation}
U_0=\frac{1}{2}\left(
                  \begin{array}{cccccccc}
                    1 & 0 & 0 & -\frac{\sqrt{6}}{2} & \frac{\sqrt{6}}{2} & 0 & 0 & 0 \\
                    0 & 2 & 0 & 0 & 0 & 0 & 0 & 0 \\
                    0 & 0 & 1 & \frac{\sqrt{6}}{2} & \frac{\sqrt{6}}{2} & 0 & 0 & 0 \\
                    -\frac{\sqrt{6}}{2} & 0 & \frac{\sqrt{6}}{2} & -1 & 0 & 0 & 0 & 0 \\
                    \frac{\sqrt{6}}{2} & 0 & \frac{\sqrt{6}}{2} & 0 & -1 & 0 & 0 & 0 \\
                    0 & 0 & 0 & 0 & 0 & 2 & 0 & 0 \\
                    0 & 0 & 0 & 0 & 0 & 0 & 2 & 0 \\
                    0 & 0 & 0 & 0 & 0 & 0 & 0 & 2 \\
                  \end{array}
                \right),
                U_1=\frac{1}{\sqrt{6}}\left(
                                         \begin{array}{cccccccc}
                                           1 & 0 & 0 & -1 & 2 & 0 & 0 & 0 \\
                                           1 & 2 & 0 & 1 & 0 & 0 & 0 & 0 \\
                                           0 & 0 & 1 & 0 & 0 & 1 & 2 & 0 \\
                                           2 & -1 & 0 & 0 & -1 & 0 & 0 & 0 \\
                                           0 & 1 & 0 & -2 & -1 & 0 & 0 & 0 \\
                                           0 & 0 & 2 & 0 & 0 & 0 & -1 & 1 \\
                                           0 & 0 & -1 & 0 & 0 & 1 & 0 & 2 \\
                                           0 & 0 & 0 & 0 & 0 & -2 & 1 & 1 \\
                                         \end{array}
                                       \right),
\end{equation}
\begin{equation}
U_0'=\frac{1}{2}\left(
                  \begin{array}{cccccccc}
                    1 & 0 & 0 & -\frac{\sqrt{6}}{2} & \frac{\sqrt{6}}{2} & 0 & 0 & 0 \\
                    0 & 1 & 0 & \frac{\sqrt{6}}{2} & \frac{\sqrt{6}}{2} & 0 & 0 & 0 \\
                     0 & 0 & 2 & 0 & 0 & 0 & 0 & 0 \\
                    -\frac{\sqrt{6}}{2} & \frac{\sqrt{6}}{2} & 0 & -1 & 0 & 0 & 0 & 0 \\
                    \frac{\sqrt{6}}{2} & \frac{\sqrt{6}}{2} & 0 & 0 & -1 & 0 & 0 & 0 \\
                    0 & 0 & 0 & 0 & 0 & 2 & 0 & 0 \\
                    0 & 0 & 0 & 0 & 0 & 0 & 2 & 0 \\
                    0 & 0 & 0 & 0 & 0 & 0 & 0 & 2 \\
                  \end{array}
                \right),
                U_1'=\frac{1}{\sqrt{6}}\left(
                                         \begin{array}{cccccccc}
                                           1 & 0 & 0 & -1 & 2 & 0 & 0 & 0 \\
                                           0 & 1 & 0 & 0 & 0 & 1 & 2 & 0 \\
                                           1 & 0 & 2 & 1 & 0 & 0 & 0 & 0 \\
                                           2 & 0 & -1 & 0 & -1 & 0 & 0 & 0 \\
                                           0 & 0 & 1 & -2 & -1 & 0 & 0 & 0 \\
                                           0 & 2 & 0 & 0 & 0 & 0 & -1 & 1 \\
                                           0 & -1 & 0 & 0 & 0 & 1 & 0 & 2 \\
                                           0 & 0 & 0 & 0 & 0 & -2 & 1 & 1 \\
                                         \end{array}
                                       \right),
\end{equation}
\begin{equation}
U_{\#}=\left(
           \begin{array}{cccccccc}
             1 & 0 & 0 & 1 & 0 & 0 & 0 & 0 \\
             0 & 1 & -1 & 0 & 0 & 0 & 0 & 0 \\
             0 & 0 & 0 & 0 & 1 & -1 & 0 & 0 \\
             1 & 0 & 0 & -1 & 0 & 0 & 0 & 0 \\
             0 & 1 & 1 & 0 & 0 & 0 & 0 & 0 \\
             0 & 0 & 0 & 0 & 1 & 1 & 0 & 0 \\
             0 & 0 & 0 & 0 & 0 & 0 & \sqrt{2} & 0\\
             0 & 0 & 0 & 0 & 0 & 0 & 0 &\sqrt{2}\\
           \end{array}
         \right),
\end{equation}
\begin{equation}
         M_{for}=\{P_f,P_r\} \ \text{and}\ M_{fin}=\{P_i:i=0\cdots,7\},\ \text{where}\ P_f=\sum_{i=0}^2|q_i\rangle\langle q_i|,\  P_r=\sum_{i=3}^7|q_i\rangle\langle q_i|\ \text{and}\ P_i=|q_i\rangle\langle q_i|.
\end{equation}

Unitary operators $U_0$ and $U_1$ are used to encode the value of $b$,  $U_0'$ and $U_1'$ are used to encode the value of $a_i$, $U_{\#}$ is used to subtract the amplitude of $|q_2\rangle$ from the amplitude of $|q_1\rangle$, respectively. The projective measurement $M_{for}$ is used to keep the quantum state in $span\{|q_0\rangle, |q_1\rangle, |q_2\rangle\}$.

\begin{Lm}
Setting the quantum state to the initial state in Step 2.2.2 can be done by a projective measurement and unitary operators.
\end{Lm}
\begin{proof}
Let the current quantum state be $|\psi\rangle=\sum_{i=0}^7 \alpha_i|q_i\rangle$, and let the projective measurement $M_{fin}$ be performed on $|\psi\rangle$. If the classical measurement result is $i$, then the resulting quantum state is $|q_i\rangle$. By applying the unitary operator $U_i=|q_0\rangle\langle q_i|$ on the state $|q_i\rangle$, the state $|q_0\rangle$ is obtained.
\end{proof}

{\bf Completeness condition}: If the input string $w\in L_{knapsack}$, then the all-powerful  prover  can make the right choice of $a_i$ and tells that to the verifier during their communications, step by step. Let $a_{j_1},a_{j_2},\cdots,a_{j_s}$ be strings selected by the prover (selection is not unique) where $j_p<j_q$ if $p<q$ and $v(b)=\sum_{i=1}^s v(a_{j_i})$.

\begin{Lm}\label{in-L-knapsack}
If the input $w\in L_{knapsack}$, then the quantum state of the verifier $V_{\varepsilon}$  will be $|q_0\rangle$ after Step 4.
\end{Lm}
\begin{proof}
Let $b=b_1b_2\cdots b_l$. The quantum state of the verifier $V_{\varepsilon}$ after the string $b$ is read is (see Appendix A for a detailed proof)
\begin{equation}\label{ec1}
    |\psi_{b}\rangle=\frac{\prod_{i=1}^l (P_fU_{b_i})|q_0\rangle}{\|\prod_{i=1}^l (P_fU_{b_i})|q_0\rangle\|}=\frac{1}{\sqrt{1+v(b)^2}}\left(
                                                                                                              \begin{array}{c}
                                                                                                                1 \\
                                                                                                                v(b) \\
                                                                                                                0 \\
                                                                                                                \vdots \\
                                                                                                                0\\
                                                                                                              \end{array}
                                                                                                            \right).
\end{equation}
 Let $a_{j_1}=a_{j_{11}}a_{j_{12}}\cdots a_{j_{1l'}}$. The quantum state of the verifier $V_{\varepsilon}$ after the string $a_{j_1}$ is read is (see Appendix A for a detail proof)

\begin{equation}\label{ec2}
    |\psi_{b\#a_{j_1}}\rangle=\frac{\prod_{i=1}^{l'} (P_fU'_{a_{j_{1i}}})|\psi_{b}\rangle}{\|\prod_{i=1}^{l'} (P_fU'_{a_{j_{1i}}})|\psi_{b}\rangle\|}=\frac{1}{\sqrt{1+v(b)^2+v(a_{j_1})^2}}\left(
                                                                                                              \begin{array}{c}
                                                                                                                1 \\
                                                                                                                v(b) \\
                                                                                                                v(a_{j_1}) \\
                                                                                                                \vdots \\
                                                                                                                0\\
                                                                                                              \end{array}
                                                                                                            \right).
\end{equation}
Afterwards, the unitary operator $U_{\#}$ and the projective measurement $M_{for}$ are performed. If the outcome of the measurement is $f$ (that is the quantum state is in $span\{|q_0\rangle, |q_1\rangle, |q_2\rangle\}$), the resulting quantum state is

\begin{equation}
|\psi_{b\#a_{j_1}\#}\rangle=\frac{P_fU_{\#}\prod_{i=1}^{l'} (P_fU'_{a_{j_{1i}}})|\psi_{b}\rangle}{\|P_fU_{\#}\prod_{i=1}^{l'} (P_fU'_{a_{j_{1i}}})|\psi_{b}\rangle\|}=\frac{1}{\sqrt{1+(v(b)-v(a_{j_1}))^2}}\left(
                                                                                                              \begin{array}{c}
                                                                                                                1 \\
                                                                                                                v(b)-v(a_{j_1}) \\
                                                                                                                0 \\
                                                                                                                \vdots \\
                                                                                                                0\\
                                                                                                              \end{array}
                                                                                                            \right).
\end{equation}

Therefore, when the verifier reaches the end of Step 4, the quantum state of the verifier is

\begin{equation}
|\psi_{w}\rangle=\frac{1}{\sqrt{1+\left(v(b)-\sum_{i=1}^s v(a_{j_i})\right)^2}}\left(
                                                                                                              \begin{array}{c}
                                                                                                                1 \\
                                                                                                                v(b)-\sum_{i=1}^s v(a_{j_i}) \\
                                                                                                                0 \\
                                                                                                                \vdots \\
                                                                                                                0\\
                                                                                                              \end{array}
                                                                                                            \right)=\left(
                                                                                                              \begin{array}{c}
                                                                                                                1 \\
                                                                                                                0 \\
                                                                                                                0 \\
                                                                                                                \vdots \\
                                                                                                                0\\
                                                                                                              \end{array}
                                                                                                            \right)=|q_0\rangle                                                                                                .
\end{equation}
\end{proof}

 The resulting quantum state in  Step 4 will be $|q_0\rangle$ with probability 1.  After  Steps 5 the
input is accepted with the probability $ P_a=2^{-k}$ and is rejected with the probability  $P_r=0$.   If  the whole process is repeated for infinitum, the accepting probability is
\begin{align}
Pr[(P^*,V_{\varepsilon})\  \text{ rejects}\  w] =\sum_{i\geq
0}(1-P_a)^i(1-P_r)^iP_a
=\frac{P_a}{P_a+P_r-P_aP_r}
=\frac{P_a}{P_a}=1.
\end{align}

{\bf Soundness condition:} If the input string $w\notin L_{knapsack}$ and $n=|w|$, no matter what choices the prover makes, there is not $I\subseteq \{1,\cdots,N\}$ such that $v(b)=\sum_{i\in I} v(a_i)$.
Suppose $a_{j_1},a_{j_2},\cdots,a_{j_s}$ are the strings selected by the prover where $j_p<j_q$ if $p<q$.

 \begin{Lm}\label{not-in-L-knapsack}
If the input string $w\not\in L_{knapsack}$, then for any prover $P^*$ the verifier $V_{\varepsilon}$ rejects $w$  after Step 4 with the probability at least $\frac{1}{2}$.
\end{Lm}
\begin{proof}
 According to the analysis in Lemma \ref{in-L-knapsack}, the quantum state at the beginning of Step 4 is
 \begin{equation}
 |\psi_{w}\rangle=\frac{1}{\sqrt{1+\left(v(b)-\sum_{i=1}^s v(a_{j_i})\right)^2}}\left(
                                                                                                              \begin{array}{c}
                                                                                                                1 \\
                                                                                                                v(b)-\sum_{i=1}^s v(a_{j_i}) \\
                                                                                                                0 \\
                                                                                                                \vdots \\
                                                                                                                0\\
                                                                                                              \end{array}
                                                                                                            \right).                                                                                              \end{equation}

 Obviously, $|v(b)-\sum_{i=1}^s v(a_{j_i})|\geq 1$.  Therefore, the probability of getting the state $|q_1\rangle$ and rejecting the input $w$ is at leat $\frac{1}{2}$.
\end{proof}

If $w\not\in L_{knapsack}$, then, according to Lemma \ref{not-in-L-knapsack}, for any prover $P^*$ the verifier $V_{\varepsilon}$ rejects the input after Step 4 with the probability
\begin{equation}
P_{r} >\frac{1}{2}
\end{equation}
and the input is accepted after Steps 5 to 7 with  the probability
\begin{equation}
P_{a}=2^{-k}.
\end{equation}
If the whole computation process is repeated indefinitely, the probability
that the verifier rejects the input is (taking into consideration that $k=1+\lceil\log_2 1/\varepsilon\rceil$)
\begin{align}
Pr[(P^*,V_{\varepsilon})\  \text{ rejects}\  w] =&\sum_{i\geq
0}(1-P_a)^i(1-P_r)^iP_r
=\frac{P_r}{P_a+P_r-P_aP_r}\\
>&\frac{P_r}{P_a+P_r}
>\frac{1/2}{2^{-k}+1/2}\geq \frac{1/2}{\varepsilon/2+1/2}
=\frac{1}{1+\varepsilon}>1-\varepsilon.
\end{align}

{\bf Time analysis:} Let $n=|w|$. The probability of getting as the outcome $f$ in the projective measurement $M_{for}$ is at least $\frac{1}{6}$ in the whole process in Figure \ref{f5}.  Hence, the probability for the verifier to reach Step 4 in one iteration is at least $\left(\frac{1}{6}\right)^n$.
 If every outcome of the projective measures $M_{for}$ in all computation steps in an iteration is $f$, then the running time from  Step 1 to Step 4 is ${\bf O}(n)$.
Therefore, the expecting running time from Step 1 to Step 4 is less than
\begin{equation}
\sum_{i=1}^{+\infty}\left(1-\left(\frac{1}{6}\right)^n\right)^{i-1}\left(\frac{1}{6}\right)^n\cdot i\cdot {\bf O}(n)=6^n\cdot {\bf O}(n).
\end{equation}
The running time of Step 5 is ${\bf O}(k)$.
The halting probability is ${\bf \Omega}\left(\frac{1}{2^{k}}\right)$  in both cases, so the expected number of repetitions of the algorithm is {\bf O}($2^{k}$) in both cases.
Hence the expected running time of the QAM $(P,V_{\varepsilon})$ is
{\bf O}$( 2^{k}\cdot (6^n\cdot {\bf O}(n)+{\bf O}(k)))$ and therefore {\bf O}$\left(  \frac{1}{\varepsilon}(n 6^n+\log\frac{1}{\varepsilon})\right)$.
 \end{proof}

 Dwork and Stockmeyer \cite{DwS92} proved that the set of languages recognized by AM(2PFA) is a proper subset of P. However, in the previous theorem we prove that the language $L_{knapsack}$, which is NP complete, can be recognized by QAM(2QCFA). This is another example that indicates QAM(2QCFA) are more powerful than AM(2PFA).

\subsection{Recognition of other languages}
 In this subsection, we sketch the proofs for the following languages, where $\Sigma=\{a,b\}$:
 \begin{equation}\label{L_1}
    L_1=\{w\,|\, \exists s,t,u,v\in \Sigma^*,\ w=sbt=ubv,|s|=|v|\},
 \end{equation}
  \begin{equation}\label{L_2}
    L_2=\{w\,|\, \exists s,t,u,v\in \Sigma^*,\ w=sat=ubv,|s|=|v|\},
 \end{equation}
  that they can be recognized by QAM(2QCFA).
  Freivalds et al. \cite{Fre10} proved that the language $L_1$ is nonstochastic, whereas $L_2$, the set of nonpalindromes, is stochastic.
  We prove that these languages can be recognized by QAM(2QCFA) in polynomial expected running time.

  \begin{Th}\label{L-1}
For any $\varepsilon<1/2$, there exists a verifier-prover pair
$(P,V_{\varepsilon})$ of a QAM(2QCFA) that can recognize $L_{1}$ with one-sided
error $\varepsilon$ in the expected running time
$\textbf{O}\left(\frac{1}{\varepsilon}\left(n^4+n^2\log{\frac{1}{\varepsilon}}\right)\right)$, where $n$ is the length of the input.
\end{Th}
\begin{proof} If $w\in L_1$, then $w$ can be of one of the following three types
\begin{enumerate}
\item [Type 1.] $w=sbt=ubv$, where $|s|=|u|=|v|$.
\item [Type 2.] $w=sbt=ubv$, where $|s|>|u|$. If $s=ubx$, where $x\in\Sigma^*$, then $w=ubxbt$ and $|u|=|w|-|v|-1=|w|-|s|-1=|t|.$
\item [Type 3.] $w=sbt=ubv$, where $|s|<|u|$. If $u=sby$, where $y\in\Sigma^*$, then $w=sbybv$ and $|s|=|v|.$
\end{enumerate}
Therefore, $L_1=\{w\,|\, \exists u,v,z\in\Sigma^*, \ w=ubv \ \text{or}\ w=ubzbv,\ \text{where}\  |u|=|v|\}$.

By virtue of  the method of the proof of Theorem \ref{L-middle}, the quantum state of the verifier will be rotated by an angle $\alpha$ every time a symbol in $u$ is scanned. The prover just tells the verifier the right time to stop rotation and of which type the input is. After the verifier checks the input is of the form $b$ or $bzb$, the prover tells the verifier the right position to resume rotation of the quantum state by the angle  $-\alpha$. The rest of the proof is similar to the one in Theorem \ref{L-middle}.

\end{proof}

  \begin{Th}\label{L-2}
For any $\varepsilon<1/2$, there exists a verifier-prover pair
$(P,V_{\varepsilon})$ of a QAM(2QCFA) that recognizes the language $L_{2}$ with one-sided
error $\varepsilon$ in the expected running time
$\textbf{O}\left(\frac{1}{\varepsilon}\left(n^4+n^2\log{\frac{1}{\varepsilon}}\right)\right)$, where $n$ is the length of input.
\end{Th}
\begin{proof}

Suppose $w\in L_2$. The input string $w$ can be of one of the following types
\begin{enumerate}
  \item [Type 1.] $w=sat=ubv$ where $|s|>|u|$. Let $s=ubx$ where $x\in\Sigma^*$, then $w=ubxat$ and $|u|=|w|-|v|-1=|w|-|s|-1=|t|.$
  \item [Type 2.] $w=sat=ubv$ where $|s|<|u|$. Let $u=say$ where $y\in\Sigma^*$, then $w=saybv$ and $|s|=|v|.$
\end{enumerate}
Therefore, $L_2=\{w\,|\,\exists u,v,z\in\Sigma^*, \ w=uazbv \ \text{or}\ w=ubzav,\ \text{where}\  |u|=|v|\}$.

The proof method of Theorem \ref{L-middle} will be employed. The prover will tell the verifier two positions -- the position of the symbol after $u$ and the position of the symbol before $v$. The verifier rotates its quantum state by an angle $\alpha$ every time it scans a symbol in $u$ and then checks whether symbols in these two positions, pointed out by the prover, are different.
If they are different, the quantum state of the verifier is rotated by an angle $-\alpha$ every time a symbol in $v$ is scanned. Otherwise, the input string is rejected.
The rest of the proof is similar to that in Theorem \ref{L-middle}.

\end{proof}

\section{Concluding remarks}

We have explored quantum interactive proof systems with 2QCFA verifiers and classical communication. We have focused on the public coin version of the interactive proof systems~--~namely QAM(2QCFA). We have showed that
QAM(2QCFA) are more powerful than their classical counterparts AM(2PFA). In particular, we have shown a number of  specific results demonstrating that: (1) The language $L_{middle}=\{xay\,|\, x,y \in \Sigma^*, \Sigma=\{a,b\}, |x|=|y|\}$ can be recognized by QAM(2QCFA) in a polynomial expected running time, but cannot be recognized by AM(2PFA) in polynomial
expected running time. (2) The language $L_{mpal}=\{xax^R\,|\, x\in \{a,b\}^* \}$
can be recognized by QAM(2QCFA) in an exponential expected running time, but cannot be recognized by AM(2PFA) at all.
(3) The 0~-~1 knapsack language can be recognized by QAM(2QCFA) in an exponential expected running time.

A related open problem, first mentioned in \cite{DwS92}, is whether AM(ptime-2PFA)  are
more powerful than 2PFA(ptime). Our attempts (and that of others) to show that
AM(ptime-2PFA) are indeed more powerful failed so far. However, the results of this paper indicate,
that the answer is likely positive when it comes to the quantum case.
Indeed, the language $L_{middle}$ can be  recognized by QAM(ptime-2QCFA),
but it seems that it can not be recognized by 2QCFA. That would
imply that QAM(ptime-2QCFA) are more powerful than 2QCFA.

\section*{Acknowledgements}

The authors would like to thank the referees for helpful suggestions to improve the presentation of the paper. This work is supported in part by the National
Natural Science Foundation of China (Nos. 61272058, 61073054) and
supported in part by Employment of Newly Graduated Doctors of Science for Scientific Excellence project/grant (CZ.1.07./2.3.00/30.0009)  of Czech Republic, and the project
of the FCT PEst-OE/EEI/LA0008/2013.

\section*{Appendix A: Proofs of Equalities \ref{ec1} and  \ref{ec2}}\label{details}

We prove Equality \ref{ec1} by induction on the length of string $b$. If $|b|=1$, then $b=``1"$. Therefore,
 \begin{equation}
 U_{1}|q_0\rangle=\frac{1}{\sqrt{6}}\left(
                                         \begin{array}{cccccccc}
                                           1 & 0 & 0 & -1 & 2 & 0 & 0 & 0 \\
                                           1 & 2 & 0 & 1 & 0 & 0 & 0 & 0 \\
                                           0 & 0 & 1 & 0 & 0 & 1 & 2 & 0 \\
                                           2 & -1 & 0 & 0 & -1 & 0 & 0 & 0 \\
                                           0 & 1 & 0 & -2 & -1 & 0 & 0 & 0 \\
                                           0 & 0 & 2 & 0 & 0 & 0 & -1 & 1 \\
                                           0 & 0 & -1 & 0 & 0 & 1 & 0 & 2 \\
                                           0 & 0 & 0 & 0 & 0 & -2 & 1 & 1 \\
                                         \end{array}
                                       \right)\left(\begin{array}{c}
                                                      1 \\
                                                      0 \\
                                                      0 \\
                                                      0 \\
                                                      0 \\
                                                      0 \\
                                                      0 \\
                                                      0
                                                    \end{array}
                                       \right)=\frac{1}{\sqrt{6}}\left(\begin{array}{c}
                                                      1 \\
                                                      1 \\
                                                      0 \\
                                                      2 \\
                                                      0 \\
                                                      0 \\
                                                      0 \\
                                                      0
                                                    \end{array}
                                       \right)
 \end{equation}
 and
  \begin{equation}
|\psi_{b}\rangle=\frac{P_fU_{1}|q_0\rangle}{\| P_fU_{1}|q_0\rangle\|}=\frac{1}{\sqrt{2}}\left(
                                                                                                              \begin{array}{c}
                                                                                                                1 \\

                                                                                                                1 \\
                                                                                                                0 \\
                                                                                                                \vdots \\
                                                                                                                0\\
                                                                                                              \end{array}
                                                                                                            \right)=
                                                                                                            \frac{1}{\sqrt{1+v(b)^2}}\left(
                                                                                                              \begin{array}{c}
                                                                                                                1 \\

                                                                                                                v(b) \\
                                                                                                                0 \\
                                                                                                                \vdots \\
                                                                                                                0\\
                                                                                                              \end{array}
                                                                                                            \right).
 \end{equation}

 Suppose that for $b=w$ and $|w|\geq 1$, the equality holds, that is
  \begin{equation}
|\psi_{w}\rangle=\frac{1}{\sqrt{1+v(w)^2}}\left(
                                                                                                              \begin{array}{c}
                                                                                                                1 \\
                                                                                                                v(w) \\
                                                                                                                0 \\
                                                                                                                \vdots \\
                                                                                                                0\\
                                                                                                              \end{array}
                                                                                                            \right).
 \end{equation}
 We prove the equation holds for $b=w\sigma$. If $\sigma=`0'$, then
 \begin{equation}
 U_{0}|\psi_w\rangle=\frac{1}{2}\left(
                  \begin{array}{cccccccc}
                    1 & 0 & 0 & -\frac{\sqrt{6}}{2} & \frac{\sqrt{6}}{2} & 0 & 0 & 0 \\
                    0 & 2 & 0 & 0 & 0 & 0 & 0 & 0 \\
                    0 & 0 & 1 & \frac{\sqrt{6}}{2} & \frac{\sqrt{6}}{2} & 0 & 0 & 0 \\
                    -\frac{\sqrt{6}}{2} & 0 & \frac{\sqrt{6}}{2} & -1 & 0 & 0 & 0 & 0 \\
                    \frac{\sqrt{6}}{2} & 0 & \frac{\sqrt{6}}{2} & 0 & -1 & 0 & 0 & 0 \\
                    0 & 0 & 0 & 0 & 0 & 2 & 0 & 0 \\
                    0 & 0 & 0 & 0 & 0 & 0 & 2 & 0 \\
                    0 & 0 & 0 & 0 & 0 & 0 & 0 & 2 \\
                  \end{array}
                \right)\cdot\frac{1}{\sqrt{1+v(w)^2}}\left(
                                                                                                              \begin{array}{c}
                                                                                                                1 \\
                                                                                                                v(w) \\
                                                                                                                0 \\
                                                                                                                 0 \\
                                                      0 \\
                                                      0 \\
                                                      0 \\
                                                                                                                0\\
                                                                                                              \end{array}
                                                                                                            \right)
                                       =
                                       \frac{1}{2(\sqrt{1+v(w)^2})}
                                       \left(\begin{array}{c}
                                                      1 \\
                                                      2v(w) \\
                                                      0 \\
                                                      -\frac{\sqrt{6}}{2} \\

                                                      \frac{\sqrt{6}}{2} \\
                                                      0 \\
                                                      0 \\
                                                      0
                                                    \end{array}
                                       \right)
 \end{equation}
 and
   \begin{equation}
|\psi_{w0}\rangle=\frac{P_fU_{0}|\psi_w\rangle}{\| U_{0}|\psi_w\rangle\|}=
                                                                                                          \frac{1}{\sqrt{1+(2v(w))^2}}\left(
                                                                                                              \begin{array}{c}
                                                                                                                1 \\

                                                                                                                2v(w) \\
                                                                                                                0 \\
                                                                                                                \vdots \\
                                                                                                                0\\
                                                                                                              \end{array}
                                                                                                            \right)=
                                                                                                          \frac{1}{\sqrt{1+(v(w0))^2}}\left(
                                                                                                              \begin{array}{c}
                                                                                                                1 \\

                                                                                                                v(w0) \\
                                                                                                                0 \\
                                                                                                                \vdots \\
                                                                                                                0\\
                                                                                                              \end{array}
                                                                                                            \right).
 \end{equation}
 If $\sigma=`1'$, then
  \begin{equation}
 U_{1}|\psi_w\rangle=\frac{1}{\sqrt{6}\sqrt{1+v(w)^2}}\left(
                                         \begin{array}{cccccccc}
                                           1 & 0 & 0 & -1 & 2 & 0 & 0 & 0 \\
                                           1 & 2 & 0 & 1 & 0 & 0 & 0 & 0 \\
                                           0 & 0 & 1 & 0 & 0 & 1 & 2 & 0 \\
                                           2 & -1 & 0 & 0 & -1 & 0 & 0 & 0 \\
                                           0 & 1 & 0 & -2 & -1 & 0 & 0 & 0 \\
                                           0 & 0 & 2 & 0 & 0 & 0 & -1 & 1 \\
                                           0 & 0 & -1 & 0 & 0 & 1 & 0 & 2 \\
                                           0 & 0 & 0 & 0 & 0 & -2 & 1 & 1 \\
                                         \end{array}
                                       \right)\left(
                                                                                                              \begin{array}{c}
                                                                                                                1 \\
                                                                                                                v(w) \\
                                                                                                                0 \\
                                                                                                                 0 \\
                                                      0 \\
                                                      0 \\
                                                      0 \\
                                                                                                                0\\
                                                                                                              \end{array}
                                                                                                            \right)
                                       =
                                       \frac{1}{\sqrt{6(1+v(w)^2)}}
                                       \left(\begin{array}{c}
                                                      1 \\
                                                      2v(w)+1 \\
                                                      0 \\
                                                      2-v(w) \\
                                                      v(w) \\
                                                      0 \\
                                                      0 \\
                                                      0
                                                    \end{array}
                                       \right)
 \end{equation}
 and
   \begin{equation}
|\psi_{w1}\rangle=\frac{P_fU_{1}|\psi_w\rangle}{\| U_{1}|\psi_w\rangle\|}=
                                                                                                          \frac{1}{\sqrt{1+(2v(w)+1)^2}}\left(
                                                                                                              \begin{array}{c}
                                                                                                                1 \\

                                                                                                                2v(w)+1 \\
                                                                                                                0 \\
                                                                                                                \vdots \\
                                                                                                                0\\
                                                                                                              \end{array}
                                                                                                            \right)=
                                                                                                          \frac{1}{\sqrt{1+(v(w1))^2}}\left(
                                                                                                              \begin{array}{c}
                                                                                                                1 \\

                                                                                                                v(w1) \\
                                                                                                                0 \\
                                                                                                                \vdots \\
                                                                                                                0\\
                                                                                                              \end{array}
                                                                                                            \right).
 \end{equation}
 The proof of Equality \ref{ec2} is similar.

 \end{document}